\documentclass[10pt,twocolumn,twoside]{IEEEtran}
\usepackage{}

\ifCLASSINFOpdf
\else
\fi
\usepackage{array}

\usepackage[cmex10]{amsmath}
\usepackage{bm}
\usepackage{enumerate}
\usepackage{amsfonts}
\usepackage{graphicx}
\usepackage{cite}
\usepackage{booktabs}
\usepackage{multirow}
\usepackage[linewidth=0.5pt]{mdframed}
\usepackage{threeparttable}
\usepackage{fixltx2e}
\usepackage{color}
\usepackage{xcolor}
\usepackage{subeqnarray}
\usepackage{cases}
\usepackage{makecell}
\usepackage{algorithmic}
\usepackage{algorithm}
\usepackage{bbold}
\usepackage{amssymb}
\usepackage{epstopdf}
\usepackage{amsthm}
\usepackage{algorithmic}
\usepackage{algorithm}

\theoremstyle{definition}
\newtheorem{theorem}{\textbf{Theorem}}
\newtheorem{definition}[theorem]{\textbf{Definition}}
\newtheorem{conjecture}[theorem]{\textbf{Conjecture}}
\newtheorem{proposition}[theorem]{\textbf{Proposition}}
\newtheorem{remark}[theorem]{\textbf{Remark}}
\newtheorem{lemma}[theorem]{\textbf{Lemma}}
\newtheorem{corollary}[theorem]{\textbf{Corollary}}
{ \bgroup
  \addtolength\abovedisplayshortskip{#1}
  \addtolength\abovedisplayskip{#1}
  \addtolength\belowdisplayshortskip{#1}
  \addtolength\belowdisplayskip{#1}}
{\egroup\ignorespacesafterend}
         
\ifCLASSOPTIONcompsoc
\usepackage[caption=false,font=normalsize,labelfon
\theoremstyle{definition}
\newtheorem{theorem}{\textbf{Theorem}}

\newtheorem{proposition}[theorem]{\textbf{Proposition}}

t=sf,textfont=sf]{subfig}
\else
\usepackage[caption=false,font=footnotesize]{subfig}

\fi
\newcommand{\tabincell}[2]{\begin{tabular}{@{}#1@{}}#2\end{tabular}}
\begin{document}
\bibliographystyle{IEEEtran}


\title{Topology and Admittance Estimation:\\ Precision Limits and Algorithms}

\author{
Yuxiao~Liu,~\IEEEmembership{Student Member,~IEEE,}
Ning~Zhang,~\IEEEmembership{Senior Member,~IEEE,}\\
Qingchun~Hou,~\IEEEmembership{Student Member,~IEEE,}
Audun~Botterud,~\IEEEmembership{Member,~IEEE,}
and~Chongqing~Kang,~\IEEEmembership{Fellow,~IEEE}

  



}

\maketitle


\begin{abstract}
Distribution grid topology and admittance information are essential for system planning, operation, and protection.
In many distribution grids, missing or inaccurate topology and admittance data call for efficient estimation methods.
However, measurement data may be insufficient or contaminated with large noise, which will introduce fundamental limits to the estimation accuracy.
This work explores the theoretical precision limits of the topology and admittance estimation (TAE) problem, with different measurement devices, noise levels, and the number of measurements.
On this basis, we propose a conservative progressive self-adaptive (CPS) algorithm to estimate the topology and admittance.
Results on IEEE 33 and 141-bus systems validate that the proposed CPS method can approach the theoretical precision limits under various measurement settings.
\end{abstract}
\begin{IEEEkeywords}
Distribution grids, topology estimation, admittance estimation, data-driven, Newton-Raphson method.
\end{IEEEkeywords}

\IEEEpeerreviewmaketitle


\section{Introduction}

The vast integration of distributed energy resources and electric vehicles raise both economic and security concerns to modern distribution grids~\cite{muruganantham2017challenges}.
Smart grid operations such as state estimation (SE), demand response, voltage control, and pricing are increasingly implemented at the distribution level.
However, accurate grid topology and line admittance, which are the pre-requisite for the above operations, are often unavailable in many medium and low-voltage distribution grids.
Therefore, efficient and accurate topology and admittance estimation (TAE) is essential in future distribution grids.

Many efforts have been made to estimate distribution grid topology and admittance.
Most of the works use the data collected from the advanced metering infrastructures (AMIs) or the micro-phasor measurement units ($\mu$PMUs).
Some researches that only focus on topology identification use the statistical information of voltage magnitudes, such as the covariances~\cite{bolognani2013identification}, the mutual information~\cite{weng2016distributed}, and the conditional independence test~\cite{deka2016estimating}.
Other researches address the joint estimation of topology and the line admittance, by formulating the problem as the maximum likelihood estimation~\cite{yu2017patopa,yu2018patopaem,zhang2020topology,moffat2019unsupervised,li2020learning}.
The above works often make some assumptions to simplify the problem, including uncorrelated nodal power/current injections~\cite{deka2016estimating,weng2016distributed,deka2020joint,park2020learning}, radial network topologies~\cite{bolognani2013identification,park2020learning,deka2020joint,miao2019distribution,zhao2020full,moffat2019unsupervised}, sufficient phasor measurements~\cite{yu2017patopa,moffat2019unsupervised,li2020learning,yu2018patopaem}, or accurate voltage measurements~\cite{zhang2020topology}.
Those assumptions may hold in some cases, but hinder the practical implementation under more general distribution system cases.
For instance, the power/current injections can be highly correlated because of similar electricity consumption or rooftop solar PV generation patterns~\cite{xu2019data}.
The distribution grids may contain loops or even be heavily meshed~\cite{sandraz2014energy}.
The measurement devices may not be sufficient in the distribution level, especially for the $\mu$PMUs that contain phasor information~\cite{bhela2017enhancing}.

In fact, the TAE in distribution grids is extremely challenging because of the poor measurements and the non-convexity of power flow models.
It still remains an open question whether the distribution grid is observable with some specific measurement devices and measurement precisions.
Some recent works discuss the fundamental limits of the TAE problem.
Moffat \textsl{et al.} proved that the admittance matrix cannot be estimated without any prior knowledge when the system contains some zero power injection buses~\cite{moffat2019unsupervised}.
Still, they do not consider any physical knowledge (e.g. the admittance matrix is symmetric and sparse) that could largely improve the estimation results.
Li \textsl{et al.} provided a theoretical relationship among the number of measurements, the prior knowledge, and the probability of estimation error~\cite{li2020learning}.
They proved a worst-case sample complexity for the linear graph learning task, i.e., how much data is required to guarantee a certain accuracy (with probability).
The above researches~\cite{moffat2019unsupervised,li2020learning} only address the estimation limits for linear learning tasks.
That is, the voltage and current magnitudes and angles are assumed to be available at all buses so that the estimation problem can be formulated by the linear Ohm's~law.
Grotas \textsl{et al.} proposed the lower bound for the transmission grids under the DC power flow (DCPF) setting~\cite{grotas2019power}.
The DCPF approximation simplifies the problem but introduces great error for the estimation problem.
Further, the DCPF approximation even incurs larger error in distribution grids, because of the high ratio of R/X and also requires the phasor measurements.

This work addresses a more practical setting: to derive the theoretical precision limits with different measurement devices and measurement precisions.
We first derive the Cramér-Rao Lower Bound (CRLB)~\cite{kay1993fundamentals} that can evaluate the best possible TAE precisions under given measurement devices, noise levels, and the number of data.
Then, we show that to design an efficient algorithm for the TAE problem is mathematically difficult.
The TAE problem is far more memory consuming, ill-conditioned, and non-convex compared with the traditional SE problem~\cite{abur2004power}.
Some traditionally well-behaved methods can easily diverge (the Newton's method) or suffer from an extremely slow converge speed (gradient-based methods), even under a simple 3-bus case.
Furthermore, we propose an algorithm for the TAE problem.
The method contains a first-order optimizer and a second-order optimizer in order to combine stability (or conservative) and fast convergence (or progressive) features with one method.
The method also contains a hybrid line search strategy to self-adaptively tune the weights of the first-order and second-order optimizer.
To this end, we name the proposed algorithm the conservative progressive self-adaption (CPS) algorithm.
Case studies on IEEE 33 and 141-bus systems show that the proposed CPS method can approach the theoretical precision limits under different experimental settings.

This work focuses on systems with a balanced power flow setting and Gaussian measurement noise.
The framework can be extended to address the unbalanced three-phase estimation problem and some non-Gaussian noise.
We also do not consider active injection approaches~\cite{liserre2007grid,cavraro2019inverter}, because they require sufficient active devices (e.g. smart inverters) and may not be compliant with the grid code~\cite{miao2019multi}.


In short, the contributions of this work are as follows:
\begin{enumerate}[1)]
  \item We quantify the theoretical precision limits for the distribution grids' TAE problem. The proposed method identifies precision limits given different measurement devices, noise levels, number of measurements, and prior topology knowledge. 
  \item We propose the CPS algorithm that is specially designed for the memory-heavy, ill-conditioned, and non-convex estimation problem. The method can approach the theoretical precision limits under a lack of voltage angle measurements and different levels of measurement noise.
\end{enumerate}

The remainder of this paper is organized as follows.
Section~\ref{sec_crlb} introduces how to evaluate the precision limits for the TAE problem.
Section~\ref{sec_cps1} demonstrates the proposed CPS algorithm.
Section~\ref{sec_case_study} provides case studies.
Finally, section~\ref{sec_conclusions} draws the conclusions.

\section{Precision Limits for Topology and Admittance Estimation}
\label{sec_crlb}

\subsection{Problem Formulation}
\label{sec_problem_formulation}
The TAE of a distribution grid can be formulated as a maximum likelihood problem using the AC power flow equations.
\begin{subequations}\label{eq_maxlike}
  \begin{equation}\label{eq_maxlike_a}
    \begin{aligned}
    \min _{G_{ij},B_{ij},\hat{V}_i^t,\hat{\theta}_i^t} &\sum_{t=1}^{T} 
    \sum_{i\in \!\!\mathcal{M}_{P}}\frac{(P_i^t\!-\!\hat{P}_i^t)^2}{\sigma_{P_i}^2}
    \!+\!
    \sum_{i\in \!\!\mathcal{M}_{Q}}\frac{(Q_i^t\!-\!\hat{Q}_i^t)^2}{\sigma_{Q_i}^2}\\
    \!+\!
    &\sum_{i\in \!\!\mathcal{M}_{V}}\frac{(V_i^t\!-\!\hat{V}_i^t)^2}{\sigma_{V_i}^2}
    \!+\!
    \sum_{i\in \!\!\mathcal{M}_{\theta}}\frac{(\theta_i^t\!-\!\hat{\theta}_i^t)^2}{\sigma_{\theta_i}^2},
    \end{aligned}
  \end{equation}
  \begin{equation}\label{eq_maxlike_b}
    \begin{aligned}
      \text{with }
      \hat{P}_{i}^{t} &=\hat{V}_{i}^{t} \sum_{j=1}^{N} \hat{V}_{j}^{t}\left(G_{i j} \cos \hat{\theta}_{ij}^{t}+B_{ij} \sin \hat{\theta}_{ij}^{t}\right), \\
      \hat{Q}_{i}^{t} &=\hat{V}_{i}^{t} \sum_{j=1}^{N} \hat{V}_{j}^{t}\left(G_{ij} \sin \hat{\theta}_{ij}^{t}-B_{ij} \cos \hat{\theta}_{ij}^{t}\right),
      \end{aligned}
  \end{equation}
\end{subequations}
where $G_{ij}/B_{ij}$ denote the ($i$,$j$)th element of conductance/susceptance matrix, $P_i^t/Q_i^t$ denote the active/reactive power injection measurements of bus $i$ snapshot $t$, $V_i^t/\theta_i^t$ denote the voltage magnitude/angle measurements of bus $i$ snapshot $t$, $\hat{P}_i^t/\hat{Q}_i^t/V_i^t/\theta_i^t$ denote the evaluated value of $P_i^t/Q_i^t/V_i^t/\theta_i^t$, $\sigma_{P_i}/\sigma_{Q_i}/\sigma_{V_i}/\sigma_{\theta_i}$ denote standard deviations of the corresponding measurements, and $\mathcal{M}_P/\mathcal{M}_Q/\mathcal{M}_V/\mathcal{M}_\theta$ denote the bus sets that the corresponding $P/Q/V/\theta$ measurements are available, respectively.
In the above problem~\eqref{eq_maxlike}, the estimation of $G_{ij}$ and $B_{ij}$ is equivalent to the estimation of topology and admittance, where zero values denote the disconnection of two buses and non-zero values denote the admittance.
We also recover all the state variables $V_i^t$ and $\theta_i^t$ as a by-product. 
The objective function~\eqref{eq_maxlike_a} minimizes the weighted squared loss of the data from the measurement sets $\mathcal{M}_P/\mathcal{M}_Q/\mathcal{M}_V/\mathcal{M}_\theta$.
The measurement sets can be any combinations of the buses in the distribution grid, from the universe to the empty set.
The standard deviations $\sigma_{P_i}/\sigma_{Q_i}/\sigma_{V_i}/\sigma_{\theta_i}$ reflect the precisions of different measurements.
The estimated values $\hat{P}_i^t/\hat{Q}_i^t/V_i^t/\theta_i^t$ satisfy the AC power flow equations in~\eqref{eq_maxlike_b}.
By neglecting the shunt admittance in the distribution grid~\cite{yu2017patopa,moffat2019unsupervised}, $G_{ij}$ and $B_{ij}$ become:
\begin{subequations}\label{eq_bus_line}
  \begin{equation}\label{eq_bus_line_a}
    G_{ij}=G_{ji}=-g_{ij},\text{ }G_{ii}=\sum_{j}g_{ij},
  \end{equation}
  \begin{equation}\label{eq_bus_line_b}
    B_{ij}=B_{ji}=-b_{ij},\text{ }B_{ii}=\sum_{j}b_{ij},
  \end{equation}
\end{subequations}
where $g_{ij}/b_{ij}$ denote the conductance/susceptance of line~($i$, $j$).

We then provide a general formulation of the estimation model.
We formulate all the measurements as an $M\times 1$ vector $\bm{z}$:
\begin{equation}\label{eq_measure_vec}
  \bm{z}\!=\!\!\left[\!\left\{
    \left\{\!P_i^t\right\}_{i\in\mathcal{M}_P}\!
    \left\{\!Q_i^t\right\}_{i\in\mathcal{M}_Q}\!
    \left\{\!V_i^t\right\}_{i\in\mathcal{M}_V}\!
    \left\{\!\theta_i^t\right\}_{i\in\mathcal{M}_\theta}
    \right\}_{t=1\sim T}\!\right]^T\!\!.
\end{equation}
Similarly, we define vector $\bm{\sigma}$ with each element corresponding to the standard deviation of measurement $\bm{z}$.
The estimated variables contain not only the voltage states but also the network admittances, with the $S\times 1$ state vector $\bm{x}$ formulated as:
\begin{equation}\label{eq_state_vec}
  \bm{x}\!=\!\!\left[\!
    \left\{g_{ij}\right\}_{(i,j)\in\mathcal{S}_T}\!\!
    \left\{b_{ij}\right\}_{(i,j)\in\mathcal{S}_T}\!\!
    \left\{\!
    \left\{\!\hat{V}_i^t\right\}_{i\in\mathcal{S}_B}\!\!\!
    \left\{\!\hat{\theta}_i^t\right\}_{i\in\mathcal{S}_B}\!\!
    \right\}_{t=1\sim T}\!\right]^T\!\!\!\!,
\end{equation}
where $\mathcal{S}_T$ denotes the set of possible connected buses, and $\mathcal{S}_B$ denotes the set of all the buses.
If no prior topology information is available, $\mathcal{S}_T$ contains all of the possible branches between any two buses.
From~\eqref{eq_maxlike}-\eqref{eq_state_vec}, we get the measurement model of the TAE problem:
\begin{equation}\label{eq_measure_model}
  \bm{z}=\bm{h}(\bm{x})+\bm{\epsilon},
\end{equation}
where $\bm{\epsilon}$ denotes the $M\times 1$ vector of measurement error.

\subsection{Geometric Illustration}
As shown in Section~\ref{sec_problem_formulation}, the estimation problem is an unconstrained nonlinear optimization problem, which belongs to the same category as the power system SE problem.
However, the TAE problem is much more difficult to solve such that the traditional methods used for SE problems have very poor performance on TAE problems.

\begin{figure}[ht!]
	\centering
  \begin{minipage}[ht]{3.3cm}
		\centering
		\includegraphics[width=0.8\linewidth]{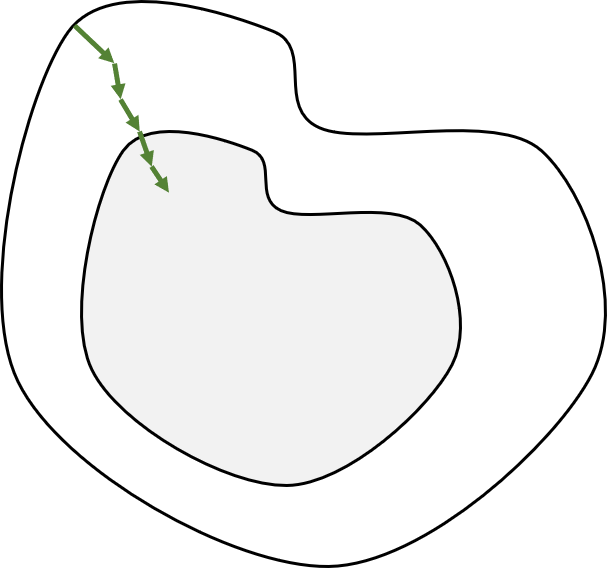}
		\centerline{\footnotesize{(a)}}
		\label{fig_opt_a}
	\end{minipage}
	\begin{minipage}[ht]{5cm}
		\centering
		\includegraphics[width=0.8\linewidth]{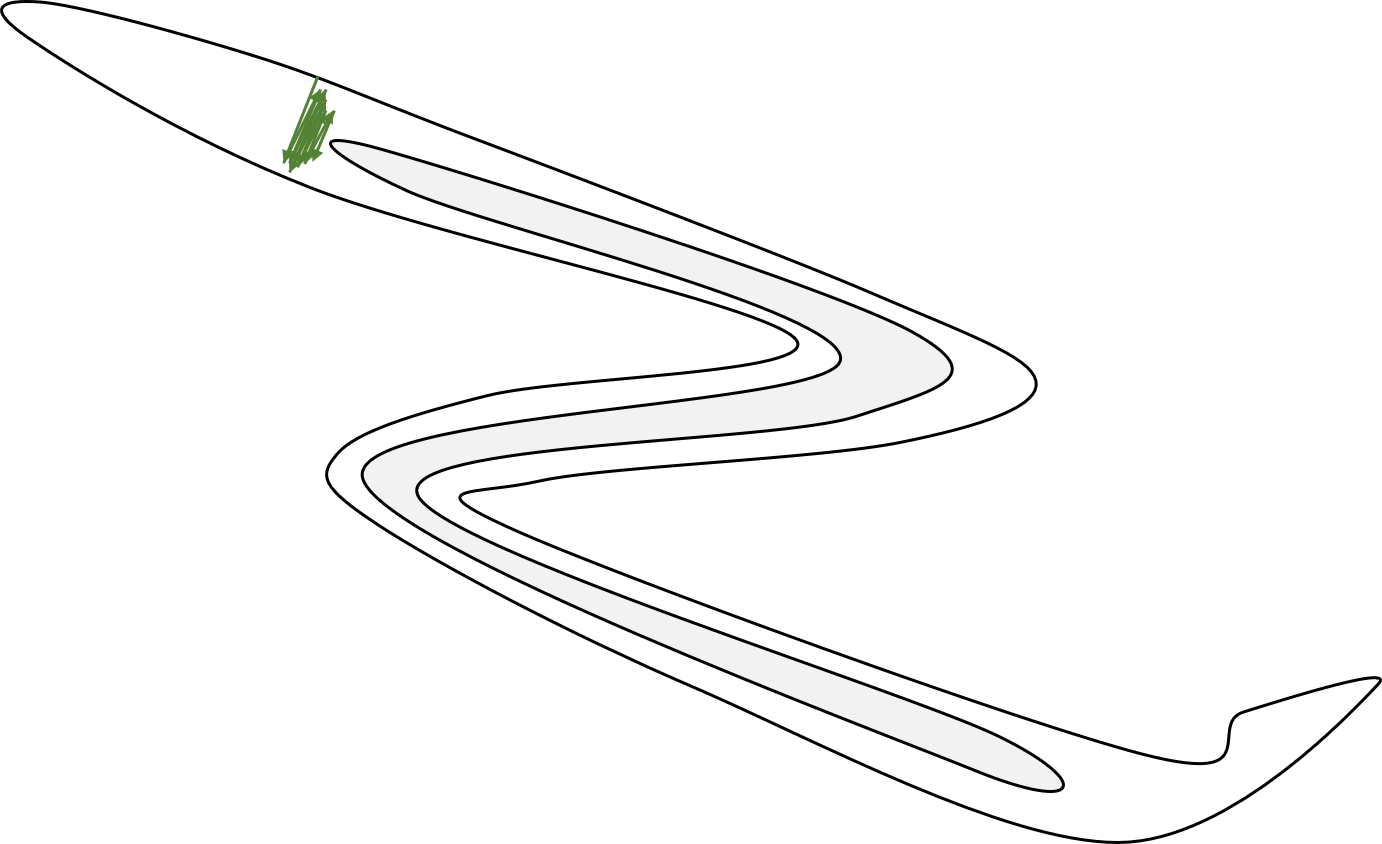}
		\centerline{\footnotesize{(b)}}
		\label{fig_opt_b}
  \end{minipage}
  \begin{minipage}[ht]{3.3cm}
		\centering
		\includegraphics[width=0.8\linewidth]{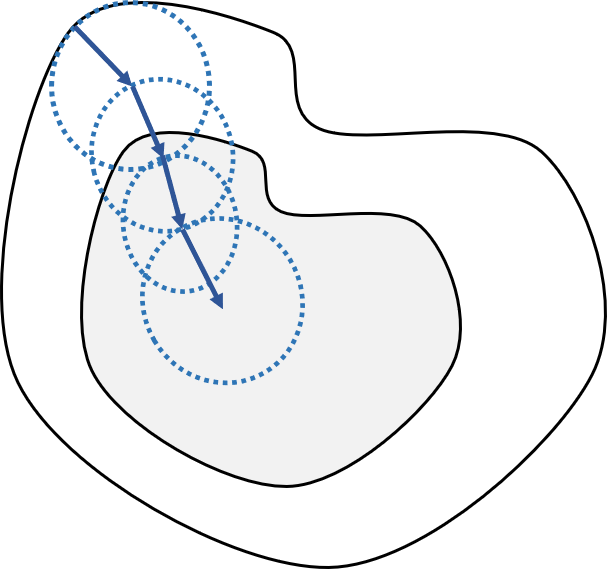}
		\centerline{\footnotesize{(c)}}
		\label{fig_opt_c}
	\end{minipage}
	\begin{minipage}[ht]{5cm}
		\centering
		\includegraphics[width=0.85\linewidth]{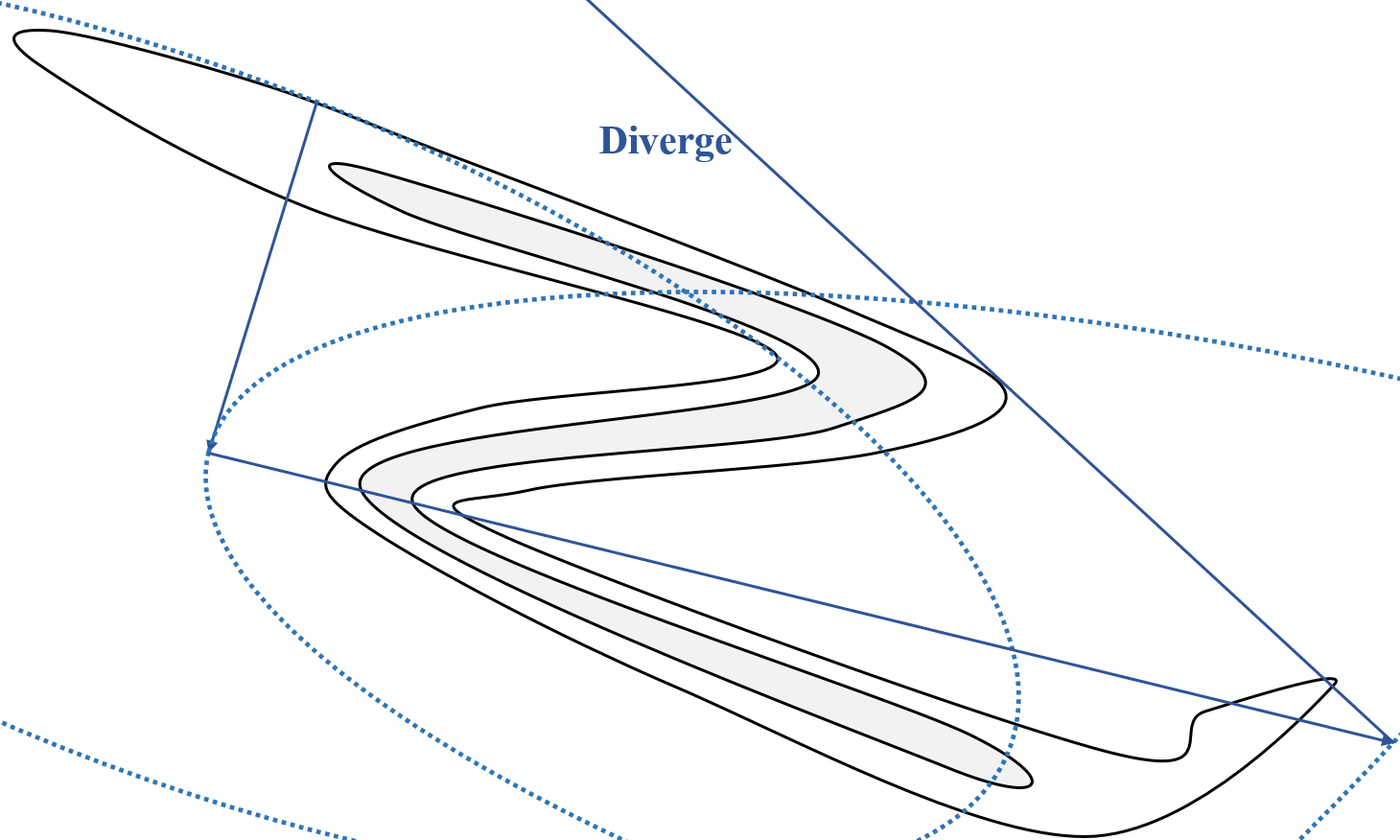}
		\centerline{\footnotesize{(d)}}
		\label{fig_opt_d}
	\end{minipage}
	\vspace{-.3cm}
  \caption{
  Geometric illustration of SE problem and TAE. 
  The 2-D space denotes the parameter space.
  The black lines are contour lines of the loss function value.
  The blue dashed lines represent the local second-order Taylor approximations.
  (a) First-order optimization for SE. 
  (b) First-order optimization for TAE. 
  (c) Second-order optimization for SE. 
  (d) Second-order optimization for TAE.
  }
	\label{fig_opt}
\end{figure}

In~\figurename~\ref{fig_opt}, we provide an illustrative example of why the TAE problem is much more difficult to solve than the SE problem:
1) the TAE problem is more non-convex because much more variables (both the model parameters and the state variables) are unknown and should be optimized.
The multiplications of more decision variables, as shown in~\eqref{eq_maxlike_b}, incur more severe non-convexity.
2) the TAE problem is more ill-conditioned than the SE problem.
The decision variables in the TAE problem have very different scales.
For example, the voltage magnitudes are often around 1 p.u. and the voltage angles are often below 1.0 rad.
However, the admittance in distribution grids can be very large (from $10^2$ to more than $10^4$), especially when the lines are short.
As a result, the parameter space can have very different scales in different directions, which makes the problem ill-conditioned.
Hence, we use a slightly non-convex contour map to represent the SE problem in~\figurename~\ref{fig_opt}~(a) and \figurename~\ref{fig_opt}~(c).
For comparison, we use a highly non-convex and ill-conditioned contour map to represent the TAE problem in~\figurename~\ref{fig_opt}~(b) and \figurename~\ref{fig_opt}~(d).

There are two types of methods to solve the unconstrained differentiable optimization problem: the first-order optimizations (the gradient-based methods) and the second-order optimizations (the Newton or quasi-Newton-based methods)~\cite{boyd2004convex}.
The first-order optimizations use the local gradient information and iteratively search the solution.
With proper step length, the first-order optimization is stable because the direction of the gradient can reduce the loss function.
Such strategy may work for SE problems as shown in \figurename{\ref{fig_opt}~(a)}, but may get trapped for the TAE problems as shown in \figurename{\ref{fig_opt}~(b)}.
On the other hand, the second-order optimizations search the solution iteratively by using the information from the second-order Taylor approximations.
This kind of method is widely adopted and proved effective in SE problems~\cite{abur2004power}.
It has a fast convergence speed for problem in~\figurename{\ref{fig_opt}~(c)}.
However, second-order optimizations can easily diverge as shown in~\figurename{\ref{fig_opt}~(d)}.






\subsection{Cramér-Rao Lower Bound}
Since the TAE problem is rather difficult, we first estimate the precision limit to better evaluate the estimation result.
We use the CRLB to estimate the precision limit of the TAE problem.
The reason for using CRLB is that it can be applied to lower bound the estimation variances of any unbiased estimator~\cite{kay1993fundamentals}, which is seldomly developed in power system analysis.
Wang \textsl{et al.} derived the CRLB for the power system SE problem~\cite{wang2018power}.
Grotas \textsl{et al.} derived the CRLB of TAE using a simplified linear DCPF model~\cite{grotas2019power}.
Damavandi \textsl{et al.}~\cite{damavandi2015robust} and Xygkis \textsl{et al.}~\cite{xygkis2016fisher} developed the Fisher Information based approach, which is closely related to the CRLB thoery, for the meter placement problem.
In this section, we derive the CRLB for the TAE problem, under the accurate non-linear AC power flow formulation. 
Compared with the existing problems~\cite{wang2018power,grotas2019power,damavandi2015robust,xygkis2016fisher}, it is a large-scale multiple snapshots non-linear estimation problem with unknown models and state variables.

\begin{theorem}[Cramér-Rao lower bound for topology and admittance estimation]
  \label{theo_CRLB}
  For the estimation problem~\eqref{eq_measure_model} with $\bm{\epsilon}\sim \mathcal{N}(\bm{0},\bm{\sigma})$, the covariance matrix $\bm{C}$ of any unbiased estimator of $\bm{x}$ satisfies:
  \begin{subequations}\label{eq_CRLB}
    \begin{equation}\label{eq_CRLB_a}
      \bm{C}-\bm{F}^{-1}\succeq \bm{0},
    \end{equation}
    \begin{equation}\label{eq_CRLB_b}
      \text{with }\bm{F}=\bm{H}^T diag(\bm{\sigma}^{-2}) \bm{H},
      \text{ }
      \bm{H}=\frac{\partial\bm{h}(\bm{x})}
      {\partial\bm{x}^T},
    \end{equation}
  \end{subequations}
  where $\bm{F}$ denotes the $S\times S$ symmetric Fisher-information matrix~\cite{kay1993fundamentals}, and $\bm{\sigma}$ denotes the $M\times 1$ vector of measurement noise standard deviations, respectively.
\end{theorem}

See Appendix \ref{appen_proof_CRLB} for the proof.
Since $\bm{C}-\bm{F}^{-1}$ is semi-definite, the variance of $\bm{x}$, which reflects the estimation precision, is lower bounded by the diagonal elements of~$\bm{F}^{-1}$:
\begin{equation}\label{eq_sigma_bound}
  \bm{\sigma_x}\geq \bm{\sigma_x^{cr}}=diag(\bm{F}^{-1}).
\end{equation}
It can be intuitively interpreted that the more information the measurement $\bm{z}$ carries about the estimated state $\bm{x}$, the lower variance the estimators can obtain.

\subsection{Fisher-information Matrix Partition}
The challenge of using CRLB for the TAE problem is that the dimension of the Fisher-information matrix is very large. 
For a N-bus distribution grid system, $\mathcal{M}_T$ contains at most $C_N^2=N(N-1)/2$ elements, $\mathcal{M}_V$ and $\mathcal{M}_\theta$ contain at most $N$ elements.
In this case, the number of state variables is $S=N(N-1)+2NT$.
Simply estimating the 33-bus distribution system using 200 snapshots of measurement data, we have to calculate the Fisher-information matrix and its inverse with $14256\times 14256$ dimension.

To address the heavy memory consumption challenge, we partition the Fisher-information matrix and calculate its inverse in a memory saving manner:
\begin{equation}\label{eq_part_matrix}
  \bm{F}=
  \left[
    \begin{matrix}
      \bm{F_{aa}}&\bm{F_{av}}\\
      \bm{F_{av}}^T&\bm{F_{vv}}
    \end{matrix}
  \right]
  =
  \left[
    \begin{matrix}
      \bm{F_{aa}}&\bm{F_{a1}}&\bm{F_{a2}}&...&\bm{F_{aT}}\\
      \bm{F_{a1}}^T&\bm{F_{11}}& & & \\
      \bm{F_{a2}}^T& &\bm{F_{22}}& & \\
      ...& & &...& \\
      \bm{F_{aT}}^T& & & &\bm{F_{TT}}\\
    \end{matrix}
  \right],
\end{equation}
where the subscript $\bm{a}$ corresponds to the admittance $\left[\left\{g_{ij}\right\}\left\{b_{ij}\right\}\right]$ in~\eqref{eq_state_vec}, the subscript $\bm{v}$ corresponds to the voltage magnitudes and angles from all the snapshots 
$[\{\{\hat{V}_i^t\}
\{\hat{\theta}_i^t\}
\}_{t=1\sim T}]$
in~\eqref{eq_state_vec}, and the subscript $\bm{T}$ corresponds to the voltage magnitudes and angles from snapshot $T$
$[\{\hat{V}_i^T\}
\{\hat{\theta}_i^T\}]$ in~\eqref{eq_state_vec}, respectively.
Note that we formulate~\eqref{eq_part_matrix} in a symmetric way.
From~\eqref{eq_part_matrix}, we can observe that the Fisher-information matrix is a sparse matrix because the voltage states in different snapshots are decoupled.
In other words, the voltage measurements in one snapshot do not contain any information about the voltage in any other snapshots.
We can take advantage of this feature and calculate the CRLB of $\left[\left\{g_{ij}\right\}\left\{b_{ij}\right\}\right]$ in a memory saving manner.
\begin{theorem}[Cramér-Rao lower bound with matrix partition]
  \label{theo_CRLB_part}
The covariance matrix $\bm{C_a}$ of any unbiased estimator of $\bm{x_a}=\left[\{g_{ij}\}_{(i,j)\in\mathcal{M}_T}\{b_{ij}\}_{(i,j)\in\mathcal{M}_T}\right]$ satisfies:
\begin{subequations}\label{eq_CRLB_part}
  \begin{equation}\label{eq_CRLB_part_a}
    \bm{C_a}-\bm{F_a}^{-1}\succeq \bm{0},
  \end{equation}
  \begin{equation}\label{eq_CRLB_part_b}
    \text{with }\bm{F_a}=\bm{F_{aa}}-\sum_{t=1}^{T}\bm{F_{at}}\bm{F_{tt}}^{-1}\bm{F_{at}}^T.
  \end{equation}
\end{subequations}
\end{theorem}
\begin{proof}
  The inverse of the Fisher-information matrix is:
  \begin{equation}
    \bm{F}^{-1}=
    \left[
      \begin{matrix}
        \bm{F_{aa}}&\bm{F_{av}}\\
        \bm{F_{av}}^T&\bm{F_{vv}}
      \end{matrix}
    \right]^{-1}
    =
    \left[
      \begin{matrix}
        \bm{I_{aa}}&\bm{I_{av}}\\
        \bm{I_{av}}^T&\bm{I_{vv}}
      \end{matrix}
    \right].
  \end{equation}
  From \textbf{Theorem~\ref{theo_CRLB}}, the CRLB of $\bm{x_a}$ is given by:
  \begin{equation}
    \bm{C_a}-\bm{F_a}^{-1}\succeq \bm{0}
    \text{, with }
    \bm{F_a}^{-1}=\bm{I_{aa}}.
  \end{equation}
  where $\bm{F_a}$ is the Schur complement of block~$\bm{F_{vv}}$:
  \begin{equation}
    \bm{F_a}=\bm{F_{aa}}-\bm{F_{av}}\bm{F_{vv}^{-1}}\bm{F_{av}^T}.
  \end{equation}
  Since $\bm{F_{vv}}$ is a block diagonal matrix, we can compute the inverse in a decoupled way and finish the proof:
  \begin{equation}
    \bm{F_{av}}\bm{F_{vv}^{-1}}\bm{F_{av}^T}=
    \sum_{t=1}^{T}\bm{F_{at}}\bm{F_{tt}}^{-1}\bm{F_{at}}^T.
  \end{equation}
\end{proof}
\begin{figure}[htb!]
	\centering
		\includegraphics[width=0.6\linewidth]{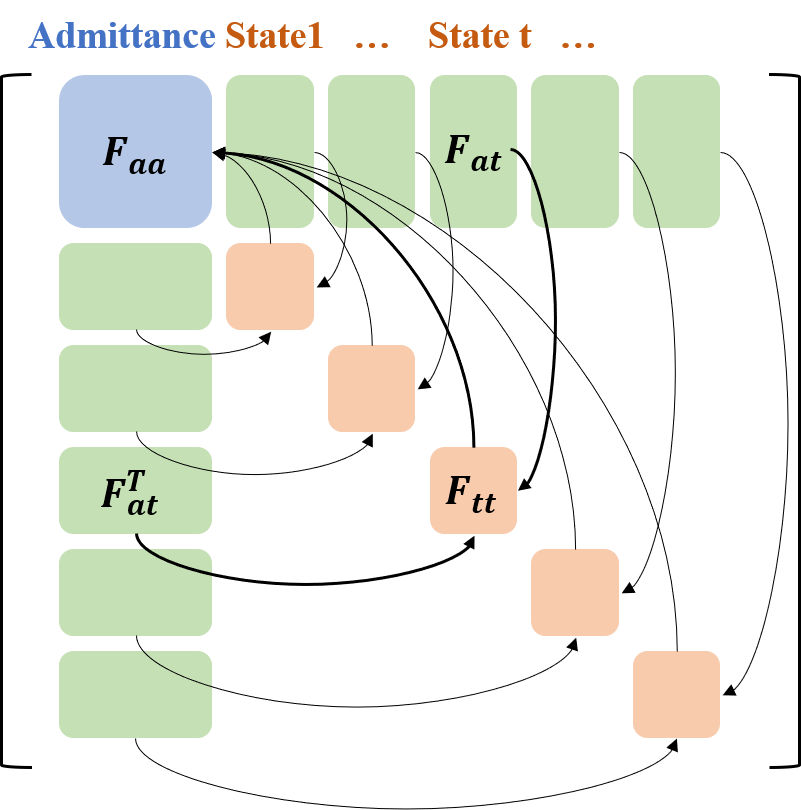}
	\caption{The structure of the partitioned Fisher-information matrix. The CRLB can be calculated in a memory saving manner.}
	\label{fig_partition}
\end{figure}

The structure of the Fisher-information matrix and way to calculate the inverse is shown in~\figurename~\ref{fig_partition}.
The required space to calculate the CRLB is largely reduced to $S'\times S'$, $S'=N(N-1)+2N$.
Since we decouple the states from different snapshots, the space complexity does not increase with the increase of snapshots.
We only need to compute the inverse of matrices with dimension of $N(N-1)\times N(N-1)$ ($\bm{F_{aa}}$) and $2N\times 2N$ ($\bm{F_{tt}}$).
Recall the 33-bus example, the dimension of the matrices are $1056\times 1056$ and $66\times 66$, much smaller than $14256\times 14256$.


\section{CPS Algorithm}
\label{sec_cps1}

\subsection{Framework of the Algorithm}

\begin{figure}[htb!]
	\centering
		\includegraphics[width=0.8\linewidth]{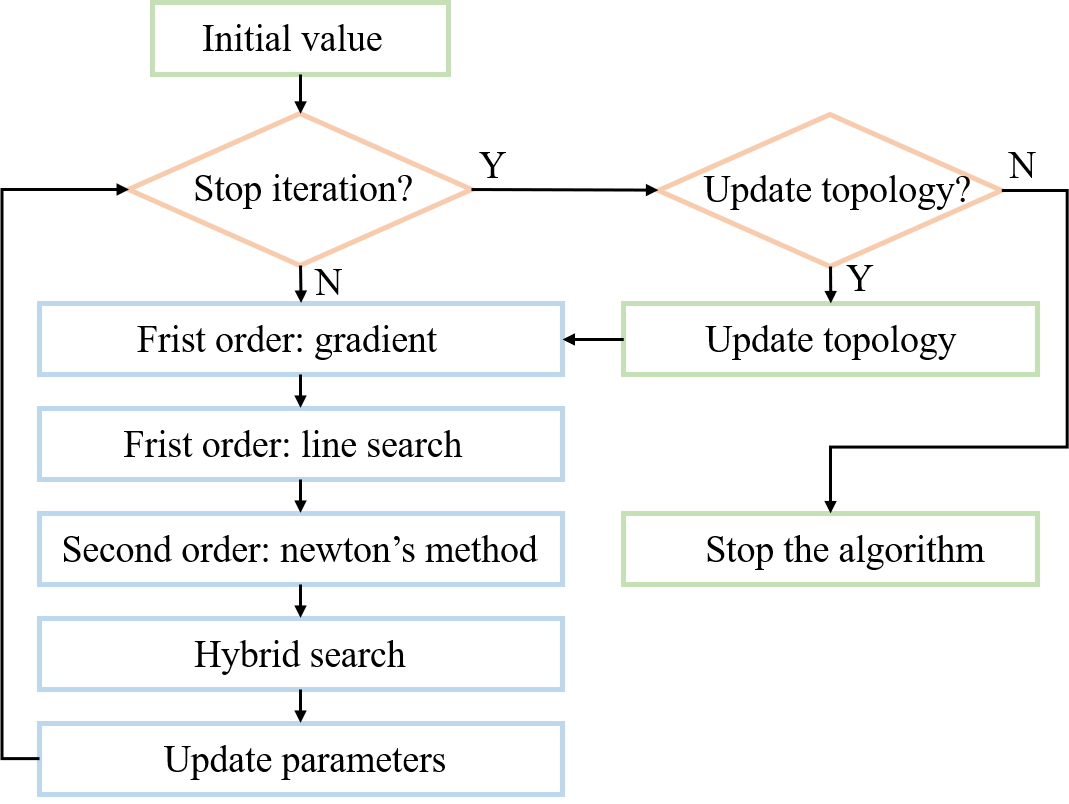}
	\caption{The framework of the CPS algorithm.}
	\label{fig_framework}
\end{figure}

The framework of the proposed CPS algorithm is shown in \figurename{\ref{fig_framework}}.
We first obtain an initial value of the admittance by evaluating a simplified model.
Then we propose an optimizer that combines the advantages of both first-order optimizations and second-order optimizations.
We take the advantage of the stability (or conservative) and the fast convergence (or progressive) from the two methods.
We also propose a hybrid search strategy to self-adaptively tune the weights of the first-order optimizer and the second-order optimizer.
Afterward, we update the topology each time after the iteration converges.
At last, we stop the algorithm when there is no need to update the topology.
The following sections will describe the details of each step in the framework.

\begin{figure}[htb!]
	\centering
		\includegraphics[width=1.0\linewidth]{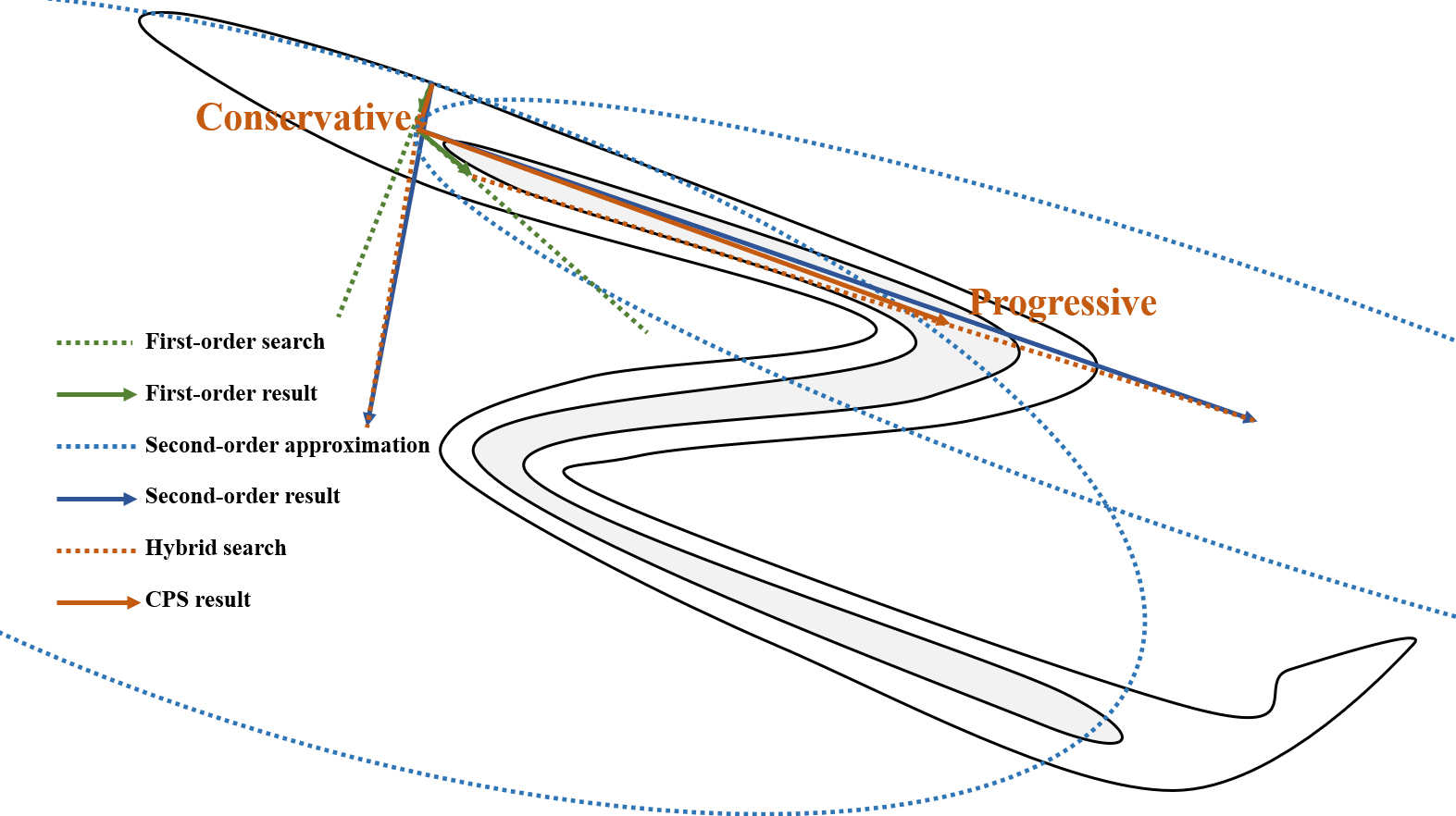}
	\caption{Geometric illustration of CPS method.}
	\label{fig_geo_CPS}
\end{figure}

In \figurename{\ref{fig_geo_CPS}}, we also provide a geometric illustration of why the CPS method can have good performance for the aforementioned highly non-convex and ill-conditioned problem.
The CPS method is conservative by taking the gradient descent direction when the second-order solution is not stable and is progressive when the second-order solution largely reduces the loss function.
In other words, the CPS method can automatically choose a proper ``mode'' between the first-order solution and the second-order solution.



\subsection{Initial Value}

The proposed CPS algorithm starts by obtaining an initial value of state vector $\bm{x}$ from the current measurements.
Note that the measurement availabilities can be a large variety and there are no general methods to obtain the initial value under all the circumstances.
Since the initial value do not require a mesh network setup or accurate parameters, we refer to~\cite{deka2020joint,park2020learning,miao2019distribution,moffat2019unsupervised,yu2017patopa,zhang2020topology} for initial value estimation under various conditions.
Even the initial values can be set as the values from the grid planning files~\cite{yu2017patopa}.
In this work, we extend the method in~\cite{zhang2020topology}, which is explained in \textbf{Appendix~\ref{appen_proof_pf_simple}}.

\subsection{First-order Optimization}
\label{sec_first-order}
The first-order optimization used in this work is modified from the adaptive moment estimation (Adam) optimizer~\cite{kingma2014adam}, which is widely applied and most favored in the training of deep learning networks.
The Adam optimizer combines the advantages of the momentum method and the root mean square propagation (RMSProp) method~\cite{kingma2014adam}.
On the one hand, Adam uses the moving average of the gradient to prevent oscillations during iterations.
On the other hand, Adam rescales the gradient by dividing the moving average of the squared gradients, so that the gradient of each variable is normalized to one or minus one.

On this basis, we improve the Adam in the way of rescaling the gradient.
Adam rescales every gradient equally to one or minus one.
However, the voltage magnitudes, angles, and admittances have very different scales and thus Adam can have a poor performance.
Instead, we first normalize the gradient and use the approximated CRLB to rescale the gradient.
The details are shown in \textbf{Algorithm~\ref{algorithm_first_order}}.

\begin{algorithm}[ht]
	\caption{The First-order Optimization}
	\label{algorithm_first_order}
	{\bf Inputs: }
	The last state vector $\bm{x^{k-1}}$ and the moment of the last iteration $\bm{m^{k-1}}$.
			
  \begin{algorithmic}[1]
  \IF{the iteration step $k=1$}
  \STATE Calculate the empirical CRLB using the initial value $\bm{x^0}$ and get $\bm{\sigma^{cr}_{x}}$. Initialize $\bm{m^0}=\bm{x^0}$. Get the mean value of $\bm{\sigma^{cr}_{x}}$ among admittances, voltage magnitudes, voltage angles, and get $[w^{cr}_a, w^{cr}_V, w^{cr}_\theta]$. 
  \STATE Set the moment ratio $\alpha$. 
  \ENDIF
	\STATE Calculate the gradient by $\bm{g^{k}}=\nabla_{\bm{x_{k-1}}} Loss$. Get the corresponding mean absolute value $[w^{g}_a, w^{g}_V, w^{g}_\theta]$.
	\STATE Normalize and rescale the gradient by $\bm{m^k}\leftarrow[\bm{g^k_a}/w^{g}_a\times w^{cr}_a, \bm{g^k_V}/w^{g}_V\times w^{cr}_V, \bm{g^k_\theta}/w^{g}_\theta\times w^{cr}_\theta]$.
	\STATE Update the moment $\bm{m^k}=\alpha\bm{m^{k-1}}+(\alpha-1)\bm{g^k}$.
	\end{algorithmic}
		
	{\bf Outputs: } 
	The moment $\bm{m^k}$ and the gradient $\bm{g^k}$.
\end{algorithm}

\subsection{Second-order Optimization}
\label{sec_second-order}
The second-order optimization used in this work is modified from the Newton's method~\cite{boyd2004convex,abur2004power}.
\begin{equation}
  \bm{x^k}=\bm{x^{k-1}}-\bm{g^k}(\bm{F^k})^{-1}.
\end{equation}
Different from the SE problem, the TAE problem is a multiple snapshots estimation problem.
The challenge for the Newton's method is the large memory consumption when calculating $\bm{g^k}(\bm{F^k})^{-1}$.
We can derive from \textbf{Theorem~\ref{theo_CRLB_part}} that the $\bm{g^k}(\bm{F^k})^{-1}$ can also be calculated in a memory saving manner.
\begin{proposition}[Second-order update with low memory consumption]
  For the Fisher-information matrix $\bm{F}$ in the form of~\eqref{eq_part_matrix} and the gradient vector in the form of
  \begin{equation}
    \bm{g}=
    \left[
    \begin{matrix}
      \bm{g_a}^T&\bm{g_v}^T
    \end{matrix}
    \right]^T
    =
    \left[
    \begin{matrix}
      \bm{g_a}^T&\bm{g_1}^T&\bm{g_2}^T&...&\bm{g_T}^T
    \end{matrix}
    \right]^T,
  \end{equation}
  $\bm{g}(\bm{F})^{-1}$ can be calculated in a memory saving manner:
  \begin{subequations}\label{eq_second_memory_saving}
    \begin{equation}
      \bm{d}=-\bm{g}(\bm{F})^{-1}=
      \left[
      \begin{matrix}
        \bm{d_a}^T&\bm{d_1}^T&\bm{d_2}^T&...&\bm{d_T}^T
      \end{matrix}
      \right]^T,
    \end{equation}
    \begin{equation}\label{eq_second_memory_saving_b}
      \text{with }
      \bm{d_a}
      =
      (\bm{F_{aa}}\!-\!\sum_{t=1}^{T}\!\bm{F_{at}}\bm{F_{tt}}^{-1}\bm{F_{at}}^T)^{-1}
      (\bm{g_a}\!-\!\sum_{t=1}^{T}\!\bm{F_{at}}\bm{F_{tt}}^{-1}\bm{g_t}),
    \end{equation}
    \begin{equation}\label{eq_second_memory_saving_c}
      \bm{d_t}=
      (\bm{F_{tt}}-\bm{F_{at}}^T\bm{F_{aa}}^{-1}\bm{F_{at}})^{-1}
      (\bm{g_t}-\bm{F_{at}}^T\bm{F_{aa}}^{-1}\bm{g_{a}}).
    \end{equation}
  \end{subequations}
  Note that for simplicity of notation, we drop the superscript of $k$ that denotes the $k$th iteration.
\end{proposition}
Recall that the TAE problem is ill-conditioned caused by different scales of admittance, voltage magnitudes and angles.
The inversion of $(\bm{F_{aa}}\!-\!\sum_{t=1}^{T}\!\bm{F_{at}}\bm{F_{tt}}\bm{F_{at}}^T)$ may incur numerical instability in~\eqref{eq_second_memory_saving_b}.
Therefore we use the Moore–Penrose inverse~\cite{ben2003generalized} instead.

\subsection{Hybrid Line Search}
Neither the first-order optimization nor the second-order optimization can provide a stable and efficient search for the highly non-convex and ill-conditioned TAE problem.
We provide a hybrid line search strategy to combine the advantages of the two optimization methods.
This strategy is inspired by the Armijo–Goldstein condition~\cite{armijo1966minimization}, a.k.a., the backtracking line search.
We modify the backtracking line search by adding the hybrid search strategy of first-order and second-order directions.

\begin{figure}[htb!]
	\centering
		\includegraphics[width=0.8\linewidth]{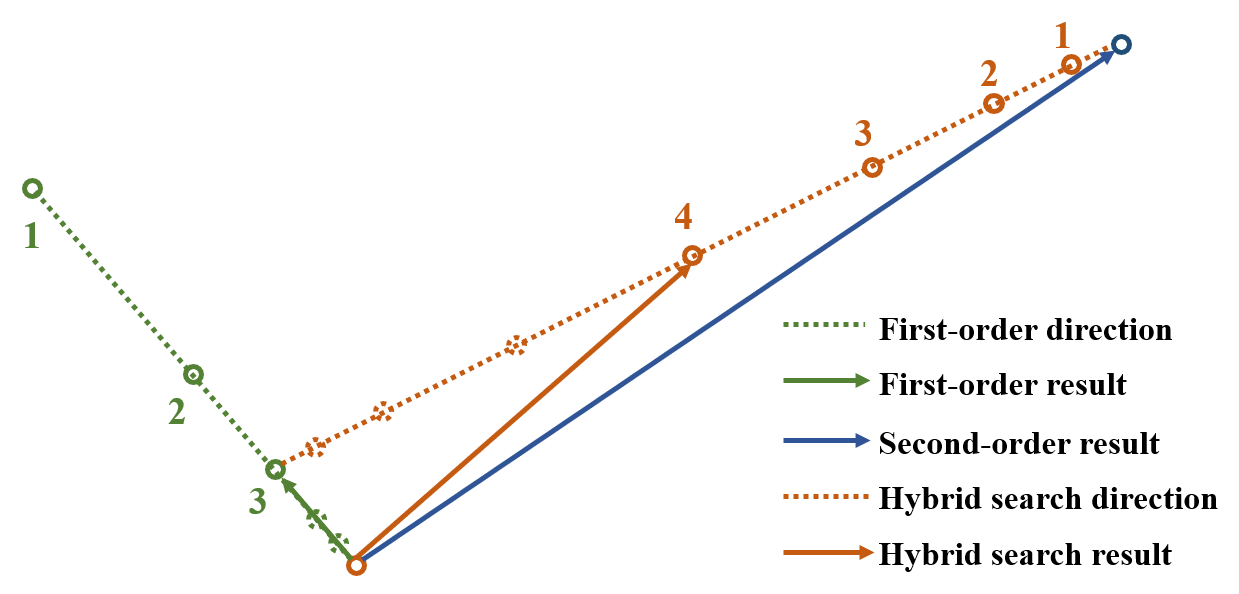}
	\caption{Geometric illustration of hybrid search.}
	\label{fig_hybrid_search}
\end{figure}

In detail, we first search for a step along the first-order direction that satisfies the Armijo–Goldstein condition.
That is, we start with a large enough step and then iteratively decrease the step size until the decrease of the loss function is large enough corresponding to the step size (or satisfies the Armijo–Goldstein condition).
It can guarantee a stable searching solution because the gradient direction will decrease the loss function with a sufficiently small step size.
We then search the step along the line from the second-order result to the first-order result.
Similarly, we start with assigning a larger weight to the second-order result and then iteratively decrease the weight of the second-order result and increase the first-order result, until the Armijo–Goldstein condition is satisfied.
In this way, we maintain a large enough proportion of second-order optimizer as long as it is stable.
The geometric illustration is shown in~\figurename{\ref{fig_hybrid_search}}.
The details of the hybrid line search are shown in \textbf{Algorithm~\ref{algorithm_hybrid_search}}.

\begin{algorithm}[ht]
	\caption{The Hybrid Line Search}
	\label{algorithm_hybrid_search}
	{\bf Inputs: }
	The last state vector $\bm{x^{k-1}}$.
			
  \begin{algorithmic}[1]
  \IF{the iteration step $k=1$}
  \STATE Set the start ratio $r_0$, the maximum ratio $r_{max}$, the incremental ratio $\beta$, and the stop threshold $\eta$.
  \ENDIF

  \STATE Do \textbf{Algorithm~\ref{algorithm_first_order}}. 
  \STATE Get $\bm{m^k}$ and $\bm{g^k}$. 
  \STATE $r\leftarrow r_0$.
  \WHILE{$r\leq r_{max}$}
  \STATE $\bm{x^k} \leftarrow \bm{x^{k-1}}+\bm{m^k}/\beta^r$. 
  \STATE $r\leftarrow r+1$.
  \IF{$Loss(\bm{x^{k-1}})-Loss(\bm{x^{k}})\leq \eta(\bm{g^k})^T\bm{m^k}/\beta^r$}
  \STATE $\bm{m^k}\leftarrow \bm{m^k}/\beta^r$. 
  \STATE \textbf{Break.}
  \ENDIF
  \ENDWHILE

  \STATE Calculate $\bm{d^k}$ from~\eqref{eq_second_memory_saving}. 
  \STATE $r\leftarrow r_0$.
	\WHILE{$r\leq r_{max}$}
  \STATE $w_1=1/(1+\beta^r)$. $w_2=\beta^r/(1+\beta^r)$.
  \STATE $\bm{x^k}\leftarrow \bm{x^{k-1}}+w_1\bm{m^k}+w_2\bm{d^k}$.
  \IF{$Loss(\bm{x^{k\!-1}})-Loss(\bm{x^{k}})\leq \eta(\bm{g^k})^T(w_1\bm{m^k}+w_2\bm{d^k})$}
  \STATE $\bm{x^{k-1}}\leftarrow \bm{x^k}$. 
  \STATE $\bm{m^k}\leftarrow \bm{m^k}/\beta^r$. 
  \STATE \textbf{Break.}
  \ENDIF
  \STATE $r\leftarrow r+1$.
  \ENDWHILE
	\end{algorithmic}
		
	{\bf Outputs: } 
	The vector of this state $\bm{x^k}$.
\end{algorithm}

\subsection{Update the Topology}

We estimate the admittances of all possible bus pairs.
Some pairs are disconnected and have zero admittances.
Therefore, we should update the topology and set the branches with small admittances to zero.
We use the method in Section III.B of~\cite{zhang2020topology} to set small admittances to zero.
Since TAE is a non-convex optimization, some non-zero admittances may approach zero temporally during the iterations.
To this end, we only update the topology when the iteration converges.
The converge criterion is:
\begin{equation}
  \max(|\bm{x^{k}}-\bm{x^{k-1}}|)\leq \gamma,
\end{equation}
where $\gamma$ is the iteration termination threshold.

It is common that some prior topology information of the distribution grids is available~\cite{li2020learning,zhang2020topology}, i.e., some bus pairs must be disconnected, so that we can set the corresponding admittances to zero during the iterations.
\section{Case Study}
\label{sec_case_study}
The power load data are from the Commission for Energy Regulation in Ireland~\cite{cer2012cer}.
We simulate the power system operational data with the aid of MATPOWER~7.0~\cite{zimmerman2010matpower}.
The simulation strategy is the same with~\cite{liu2018data}.
Additive white Gaussian noise is then added to the data.
In all our case studies, we set our hyper-parameters as: $\alpha=0.9$, $r_0=-5$, $r_{max}=20$, $\beta=5$, $\eta=0.01$, and $\gamma=10^{-5}$, respectively.

\subsection{Convergence of the Optimizer}
\begin{figure}[htb!]
	\centering
		\includegraphics[width=.6\linewidth]{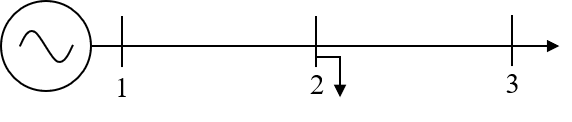}
	\caption{The 3-bus system.}
	\label{fig_case3_syst}
\end{figure}

We first use a simple 3-bus system in~\figurename{\ref{fig_case3_syst}} to compare the convergence of the proposed method with other methods.
We assume the topology of the 3-bus system is known and only estimates the admittance.
We estimate the admittance with 50 snapshots of data, with $P$, $Q$, $V$ measurements under 0.1\% noise.

\begin{figure}[htb!]
	\centering
		\includegraphics[width=1\linewidth]{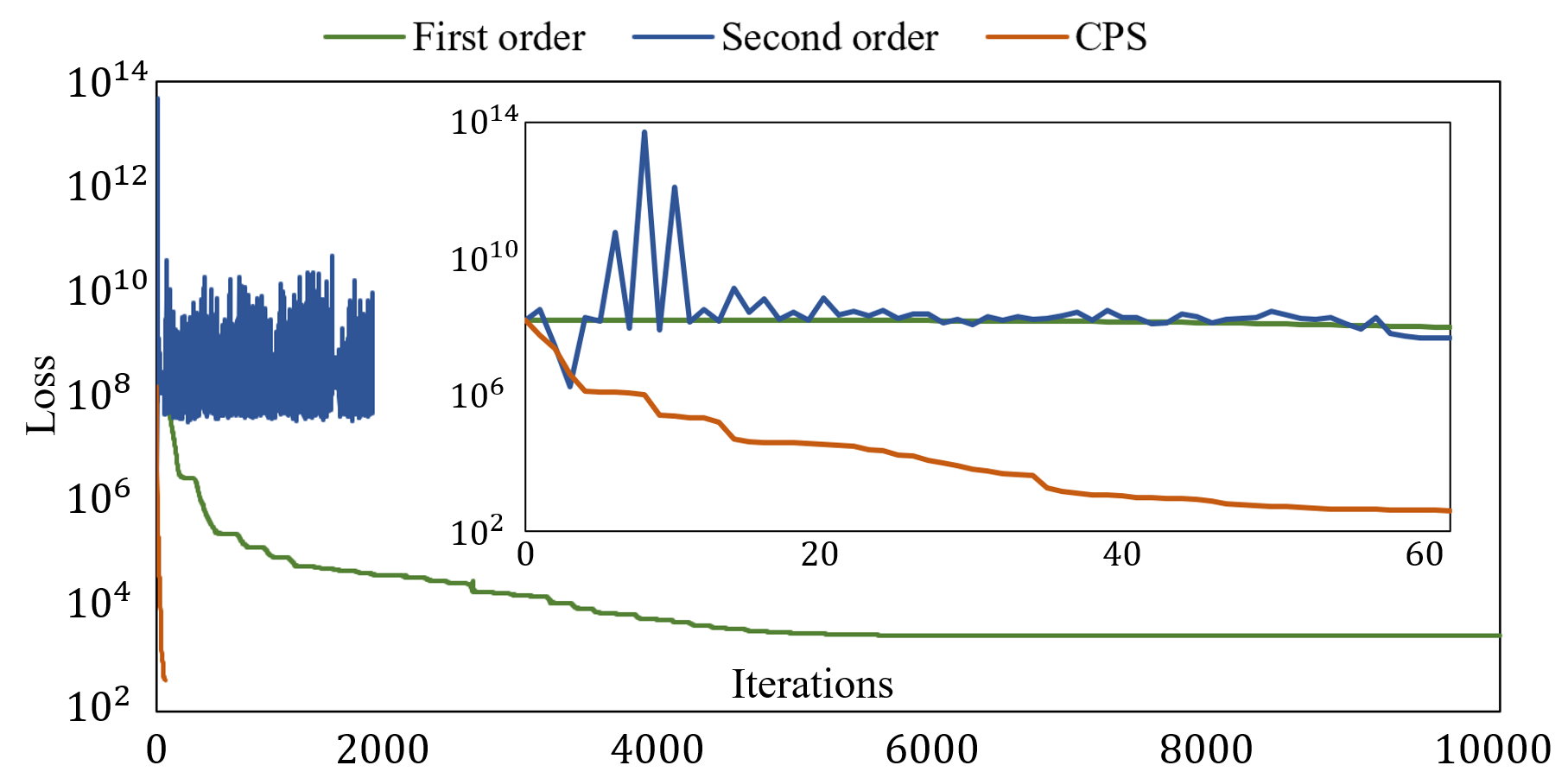}
	\caption{The loss of 3-bus system during the iterations.}
	\label{fig_case3_loss}
\end{figure}

We compare the losses during iterations with three methods: the first-order optimization in Section~\ref{sec_first-order}, the second-order optimization in Section~\ref{sec_second-order}, and the proposed CPS algorithm.
As shown in \figurename{\ref{fig_case3_loss}}, both the first-order and the second-order optimization have very poor performances even with a very simple 3-bus system.
The first-order optimization has a very slow converge speed.
The loss is over $10^3$ and fails to decrease after 10000 iterations.
The second-order optimization is not stable and the loss values oscillate in the range of $10^8\sim 10^{10}$.
However, the proposed CPS method converges in only 63 iterations with loss approaches to $3\times 10^2$.
The performance of the compared three methods can be well explained by the illustrations in \figurename{~\ref{fig_opt}} and \figurename{\ref{fig_geo_CPS}}.
The proposed CPS method shows extraordinary performance in the TAE problem.

\subsection{The Bound Attainability}


\begin{figure*}[htb!]
	\centering
    \includegraphics[width=1\linewidth]{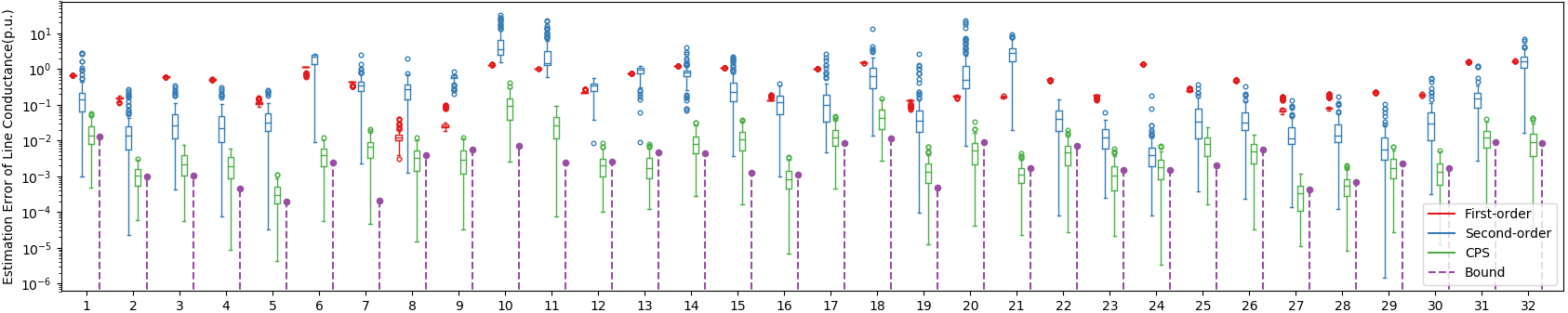}
    \centerline{\footnotesize{(a)}}\\[4pt]
    \includegraphics[width=1\linewidth]{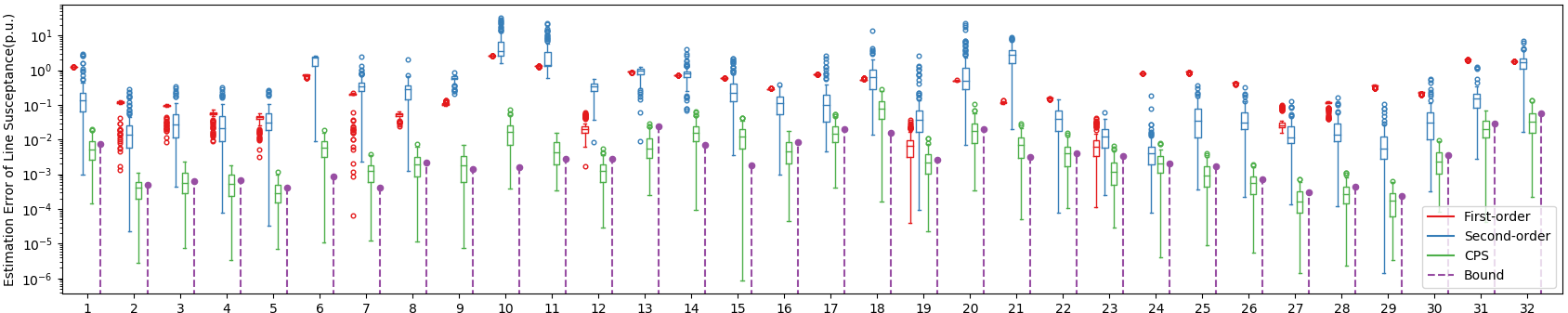}
    \centerline{\footnotesize{(b)}}
	\caption{The histogram of the admittance estimation error. (a) The error of line conductance. (b) The error of line susceptance.}
	\label{fig_case33_hist}
\end{figure*}

We then use the 12.66 kV 33-bus system~\cite{baran1989network} to demonstrate how well can the CPS method approaches the theoretical precision limits in Section~\ref{sec_crlb}.
We assume we have the $P$, $Q$, $V$, $\theta$ measurement from all the buses, with 0.1\% measurement error.
The snapshot of the measurement is 120.
We run the CPS algorithm 100 times under different randomly generated measurement noise and obtain the mean absolute error.
Under this setting, we can accurately estimate all the topologies.
The histogram of the admittance estimation error is shown in \figurename{\ref{fig_case33_hist}}.
We compare the estimation error of the first-order optimization, the second-order estimation, and the proposed CPS method.
We also demonstrate the theoretical precision limit evaluated by the method in Section~\ref{sec_crlb}.
It can be concluded that the proposed CPS method is 1$\sim$3 orders of magnitudes more accurate than the compared methods.
From the histogram of estimation error, the variance of the estimation error can approach the theoretical precision limit.



\subsection{Different Settings}

\begin{table*}[ht]
	\centering
	\renewcommand{\arraystretch}{1.1}
	\caption{
    The theoretical precision limits and the performance of CPS under different settings. 
    }
	\label{table_settings}
		\begin{tabular}{@{}lllllllll@{}}
      \toprule
      System & Snapshots & Sensors & Noise & \tabincell{l}{Prior topology\\ knowledge} &\tabincell{l}{Topology pre-\\cision limits} &\tabincell{l}{Estimated topo-\\logy precision} & \tabincell{l}{Admittance pre-\\csion limits} & \tabincell{l}{Estimated admit-\\tance precision} \\ \midrule
      33-bus & 120 & $P$, $Q$, $V$, $\theta$ & 0.1   & No & 0 & 0 & 0.264 & 0.236 \\
      33-bus & 120 & $P$, $Q$, $V$, $\theta$ & 0.2   & Yes & 0 & 0 & 0.475 & 0.521 \\
      33-bus & 120 & $P$, $Q$, $V$, $\theta$ & 0.2   & No & 3.125 & 3.125 & 0.701 & 0.746 \\
      33-bus & 120 & $P$, $Q$, $V$, $\theta$ & 0.5   & Yes & 0 & 0 & 1.188 & 1.370 \\
      33-bus & 120 & $P$, $Q$, $V$ & 0.1   & Yes & 0 & 0 & 28.711 & 29.805 \\ 
      141-bus & 200 & $P$, $Q$, $V$, $\theta$ & 0.1   & Yes & 0 & 0 & 2.346 & 2.823 \\ 
      141-bus & 200 & $P$, $Q$, $V$ & 0.01   & Yes & 0 & 0 & 21.747 & 33.625 \\ 
      \bottomrule
      \end{tabular}
\end{table*}

At last, we use the 33-bus system and the 141-bus system~\cite{khodr2008maximum} to show the precision limits and the performance of the CPS algorithm under different settings.
We compare the following different settings as shown in \textbf{Table~\ref{table_settings}}.
The experimental settings include the types of sensors, the noise level, and the prior topology knowledge.
All the errors are transformed to a scale of 100\%.
We use the ratio of wrong branches number to the true branches number as the error of topology estimation.
We use the relative geometric mean as the error of admittance estimation.
The prior topology knowledge denotes only the connectivity of existing transmission lines or tie lines are unknown.
That is, we know the disconnection of the branches that do not exist.
As can be concluded from \textbf{Table~\ref{table_settings}}, the proposed CPS method can approach the precision limits under various settings.
Still, there are some fundamental limits from the results:
1) the estimation error will largely increase (more than one degree of magnitude) without the phasor measurements;
2) the prior topology knowledge is good compensation for poor measurements;
3) the estimation of the 141-bus system with larger admittances and more buses is far more difficult than the 33-bus system and requires more accurate measurements.


\section{Conclusions}
\label{sec_conclusions}

Our work quantifies the theoretical precision limits for the distribution grids TAE problem.
We can identify the best possible estimation precisions under some specific measurement settings, i.e. different measurement devices, noise levels, number of measurements, and prior knowledge.
We also demonstrate the TAE problem is more difficult than the SE problem, in terms of the non-convexity and the ill-conditioning.
For the TAE problem, the widely applied Adam or Newton method can fail even with a very simple 3-bus system.
We propose a CPS algorithm to adaptively combine the stability of the first-order optimizer and the fast convergence of the second-order optimizer.
We also enable large scale second-order optimization by decoupling the calculation from different snapshots.
Case studies on IEEE 33 and 141-bus systems validate the proposed CPS method can approach the theoretical precision limits under different settings.
We also conclude some fundamental limits from our theory.
For example, the estimation error will largely increase without phasor measurements and can be compensated by the prior topology knowledge.
Future works will explore more theoretical relationships among the measurement devices, the measurement noise, the number of snapshots, and the prior knowledge.

\bibliography{IEEEabrv,myReference}

\clearpage


\begin{appendices}
\section{Proof of \textbf{Theorem~\ref{theo_CRLB}}}
\label{appen_proof_CRLB}
The negative log-likelihood distribution of problem~\eqref{eq_measure_model} is the multiplication of the following Gaussian distributions:
\begin{subequations}\label{eq_log_like}
  \begin{flalign}\label{eq_log_like_a}
    &
    p(\bm{z}, \bm{x})
    =
    -\ln \prod_{m=1}^M
    \frac{1}{\sqrt{2\pi}}
    \exp \left[
      - \frac{(z_m-h_m(\bm{x}))^2}{2\sigma_m^2}
    \right]
    &
  \end{flalign}
  \begin{flalign}\label{eq_log_like_b}
    &
    \qquad
    \quad
    =
    \sum_{m=1}^M
    \left[
      \frac{(z_m-h_m(\bm{x}))^2}{2\sigma_m^2}
      +
      \frac{\ln(2\pi\sigma_m^2)}{2}
    \right].
    &
  \end{flalign}    
\end{subequations}
The Fisher information matrix is defined as the expectation of the Hessian matrix of $p(\bm{z}, \bm{x})$:
\begin{subequations}\label{eq_Hessian}
  \begin{flalign}\label{eq_Hessian_a}
    &
    \bm{F}
    =
    \mathbb{E}_{\bm{x}}
    \left[
      \frac{\partial^2}{\partial\bm{x}^T\partial\bm{x}}
      p(\bm{z}, \bm{x})
    \right]
    &
  \end{flalign}
  \begin{flalign}\label{eq_Hessian_b}
    &
    \quad
    =
    \mathbb{E}_{\bm{x}}
    \left[
      \frac{\partial}{\partial\bm{x}^T}
      \sum_{m=1}^M
      \left(
        -\frac{z_m-h_m(\bm{x})}{\sigma_m^2}
        \times
        \frac{\partial h_m(\bm{x})}{\partial\bm{x}}
      \right)
    \right]
    &
  \end{flalign}    
  \begin{flalign}\label{eq_Hessian_b}
    &
    \quad
    =
    \mathbb{E}_{\bm{x}}
    \left[
      \sum_{m=1}^{M}
      \frac{1}{\sigma_m^2}\times
      \frac{\partial^2 h_m(\bm{x})}{\partial\bm{x}^T\partial\bm{x}}
      \left[
      1-(z_m-h_m(\bm{x}))
      \right]
    \right].
    &
  \end{flalign}    
\end{subequations}
Since the Gaussian noise distribution is symmetric to zero, the expectation of $z_m-h_m(\bm{x})$ over the whole distribution is zero.
Therefore, we have:
\begin{equation}
  \bm{F}=\frac{\partial h^T(\bm{x})}{\partial \bm{x}}
  \times
  \frac{1}{\bm{\sigma}_m^2}
  \times
  \frac{\partial h(\bm{x})}{\partial \bm{x}^T}.
\end{equation}
According to the CRLB~\cite{kay1993fundamentals}, the covariance matrix is lower bounded by the inverse of Fisher information matrix.


\section{Method to obtain the initial value}
\label{appen_proof_pf_simple}

We extend the method in~\cite{zhang2020topology} to obtain the initial value, assuming there are power injection measurements $(P_i^t, Q_i^t)$ and voltage magnitude measurements $V_i^t$.
The reason for this setting is the phasor measurements $\theta_i^t$ are usually not available in distribution grids.
We also assume the topology is radial at the initial stage, which can be relaxed in the following iteration stages.
The initial value of $\bm{x}$ begins by setting the voltage magnitudes measurement as the initial value of $[\{\hat{V}_i^t\}]$.
Then, we estimate an initial value of $\left[\left\{g_{ij}\right\}\left\{b_{ij}\right\}\right]$ from the $P_i^t$, $Q_i^t$, and $V_i^t$ measurements.
At last, we obtain the initial value of $[\{\hat{\theta}_i^t\}]$ by the DC power flow calculation.

The main difficulty is how to estimate the $\left[\left\{g_{ij}\right\}\left\{b_{ij}\right\}\right]$ without the voltage angle measurements.

We start by estimating a radial topology using the voltage magnitudes by a tree structure topology construction.
We use the average value of the voltage magnitudes to obtain the root-leaf relationship of the distribution grids.
We assume the voltage magnitudes are decreasing from the root to the leaf nodes.
We can only use the measurements at the night to avoid the possible bidirectional power flow caused by distributed PV generation.
We then calculate the moving average of the voltage magnitudes to decrease the influence of measurement noise:
\begin{equation}
  \hat{V}_i^1=\frac{1}{l}\sum_{t=1}^{l}V_i^t.
\end{equation}
Lastly, we use the correlation coefficient to construct the topology.
It is under the assumption that the voltage magnitudes of the connected buses are highly correlated~\cite{weng2016distributed,deka2020joint}.
See \textbf{Algorithm~\ref{algorithm_tree_topology}} for details.
Note that the proposed method may not obtain an accurate result because the assumptions may not be satisfied in practice.
However, \textbf{Algorithm~\ref{algorithm_tree_topology}} only provides an initial value and the inaccuracy can be corrected in the following iterations.
Further, some recent works also provide some topology estimation methods~\cite{weng2016distributed,deka2020joint,zhao2020full} and can be used to obtain the initial topology.

\begin{algorithm}[ht]
	\caption{Tree Structure Topology Construction}
	\label{algorithm_tree_topology}
	{\bf Inputs: }
	The voltage magnitude measurements of different buses $[V_1,...,V_i,...,V_N]$.
			
  \begin{algorithmic}[1]
  \STATE Initialize the topology set $\mathcal{S_T}=\emptyset$. Sort the average voltage magnitudes in a descending order $[V_{d1},...,V_{di},...,V_{dN}]$.
  \STATE Get the $l$ moving average of the voltage magnitudes $[\hat{V}_{d1},...,\hat{V}_{di},...,\hat{V}_{dN}]$.
  \FOR {$i=2\sim N$} 
    \STATE Get the bus $j$ with maximum correlation coefficient $\arg\max_{j\in 1\sim i-1}Corr(\hat{V}_{di}, \hat{V}_{dj})$.
    \STATE $\mathcal{S_T}\leftarrow({di}, {dj})$.
  \ENDFOR
	\end{algorithmic}
		
	{\bf Outputs: } 
	The topology set $\mathcal{S_T}$.
\end{algorithm}

We then estimate line admittance by introducing an approximate power flow formulation without voltage angles.

\begin{theorem}[Phasor-free power flow]
  \label{theo_pf_simple}
  The power flow equations can be approximated as~\eqref{eq_pf_simple} under the assumption that $\alpha_{ij}=P_{ij}/Q_{ij}$ is a constant and $\sin\theta_{ij}\approx\theta_{ij}$.
  \begin{subequations}\label{eq_pf_simple}
    \begin{equation}\label{eq_pf_simple_a}
      P_{i}/V_{i}=(V_i-V_j)/z^{\#}_{ij},
    \end{equation}
    \begin{equation}\label{eq_pf_simple_b}
      Q_{i}/V_{i}=(V_i-V_j)/(z^{\#}_{ij}\alpha{ij}),
    \end{equation}
  \end{subequations}
  with the augmented impedance $z^{\#}_{ij}$ defined as:
  \begin{equation}
    z^{\#}_{ij}=
    \frac{g_{ij}-b_{ij}/\alpha_{ij}}{g_{ij}^2+b_{ij}^2}.
  \end{equation}
\end{theorem}

\begin{proof}
The branch flow equations are formulated as:
\begin{subequations}\label{eq_branch_flow}
  \begin{equation}\label{eq_branch_flow_a}
    \begin{aligned}
      P_{i}/V_{i}
      &=(V_i-V_j\cos\theta_{ij})g_{ij}-V_j\sin\theta_{ij}b_{ij}\\
      &\approx(V_i-V_j\cos\theta_{ij})g_{ij}-V_j\theta_{ij}b_{ij},
    \end{aligned}
  \end{equation}
  \begin{equation}\label{eq_branch_flow_b}
    \begin{aligned}
      Q_{i}/V_{i}
      &=-(V_i-V_j\cos\theta_{ij})b_{ij}-V_j\sin\theta_{ij}g_{ij}\\
      &\approx-(V_i-V_j\cos\theta_{ij})b_{ij}-V_j\theta_{ij}g_{ij},
    \end{aligned}
  \end{equation}
\end{subequations}
where we use the assumption $\sin\theta_{ij}\approx\theta_{ij}$.
$\eqref{eq_branch_flow_a}\times b_{ij}+\eqref{eq_branch_flow_b}\times g_{ij}$ is:
\begin{subequations}
  \begin{equation}
    \frac{P_{ij}b_{ij}}{V_i}+\frac{Q_{ij}g_{ij}}{V_i}=-(b_{ij}^2+g_{ij}^2)\theta_{ij},
  \end{equation}
  \begin{equation}\label{eq_theta}
    \theta_{ij}=-\frac{1}{V_i}
    \left(
      \frac{b_{ij}P_{ij}}{b_{ij}^2+g_{ij}^2}+\frac{g_{ij}Q_{ij}}{b_{ij}^2+g_{ij}^2}
    \right).
  \end{equation}
\end{subequations}
Substitute~\eqref{eq_theta} into~\eqref{eq_branch_flow_a} and substitute~\eqref{eq_theta} into~\eqref{eq_branch_flow_b}:
\begin{subequations}\label{eq_pf_simple_appen}
  \begin{equation}\label{eq_pf_simple_appen_a}
    \frac{P_{ij}}{V_i}=\frac{g_{ij}^2+b_{ij}^2}{g_{ij}-\alpha_{ij}^{-1}b_{ij}}(V_i-V_j),
  \end{equation}
  \begin{equation}\label{eq_pf_simple_appen_b}
    \frac{Q_{ij}}{V_i}=-\frac{g_{ij}^2+b_{ij}^2}{b_{ij}-\alpha_{ij} g_{ij}}(V_i-V_j),
  \end{equation}
\end{subequations}
where $\alpha_{ij}=P_{ij}/Q_{ij}$.
\end{proof}

Afterwards, we can estimate the value of $1/z^{\#}_{ij}$ and $1/(z^{\#}_[ij]\alpha_{ij})$ by the least squares estimation and set them as the initial value of $g_{ij}$ and $b_{ij}$, respectively.
\begin{subequations}\label{eq_least_squares}
  \begin{equation}\label{eq_least_squares_a}
    g_{ij}=(P_i/V_i)^T(V_{i}-V_{j})/[(V_{i}-V_{j})(V_{i}-V_{j})^T)],
  \end{equation}
  \begin{equation}\label{eq_least_squares_b}
    b_{ij}=(Q_i/V_i)^T(V_{i}-V_{j})/[(V_{i}-V_{j})(V_{i}-V_{j})^T)].
  \end{equation}
\end{subequations}

\end{appendices}

\ifCLASSOPTIONcaptionsoff
  \newpage
\fi

\end{document}